\documentclass[runningheads,a4paper,10pt]{llncs}
\usepackage{geometry}
\usepackage{amssymb}
\usepackage{graphicx}
\usepackage{epstopdf}

\usepackage{epsfig}
\usepackage{subfigure}
\usepackage{amsmath}
\usepackage{algorithmic}
\usepackage{graphics}
\usepackage{amssymb}
\usepackage{alltt}
\usepackage{float}
\usepackage{array}
\usepackage{xspace}
\usepackage{footmisc}
\usepackage{comment}
\usepackage{algorithm}
\usepackage{algorithmic}

\usepackage[para*]{manyfoot}

\DeclareNewFootnote[para]{B}[roman]

\usepackage{url}
\usepackage[justification=centering]{caption}
\usepackage{graphics}

\newcounter{RomanNumber}
\newcommand{\MyRoman}[1]{\setcounter{RomanNumber}{#1}\Roman{RomanNumber}}
\newcommand{\qdtree}{Quadtree\xspace}
\newcommand{\pyr}{Pyramid\xspace}
\newcommand{\gst}{Grid structure\xspace}
\newcommand{\gtree}{Grid-Tree\xspace}
\newcommand{\gtrees}{Grid-Trees\xspace}
\newcommand{\sgtree}{SG-Tree\xspace}

\newcommand{\esg}{SSG-Trees\xspace}
\newcommand{\ssg}{SSG-Trees\xspace}
\newcommand{\sg}{SG-Tree\xspace}

\newcommand{\logs}{Log Structure\xspace}

\newcommand{\window}{Sliding Window\xspace}

\newcommand{\st}{\textit{spatio-textual}\xspace}
\newcommand{\rtree}{R-tree\xspace}

\newcommand{\irtree}{IR$^2$-tree\xspace}

\usepackage{url}
\urldef{\mailsa}\path|{alfred.hofmann, ursula.barth, ingrid.haas, frank.holzwarth,|
\urldef{\mailsb}\path|anna.kramer, leonie.kunz, christine.reiss, nicole.sator,|
\urldef{\mailsc}\path|erika.siebert-cole, peter.strasser, lncs}@springer.com|
\urldef{\mailsd}\path|{cyzhang,  leizhu, jlong, hfang, zhbalexander}@csu.edu.cn, w.yu3@aston.ac.uk|
\newcommand{\keywords}[1]{\par\addvspace\baselineskip
\noindent\keywordname\enspace\ignorespaces#1}

\begin{document}

\mainmatter  

\title{Efficient Top K Temporal Spatial Keyword Search}

\titlerunning{Efficient Top K Temporal Spatial Keyword Search}

%
%
\author{}
\authorrunning{Efficient Top K Temporal Spatial Keyword Search}
\institute{}
\author{Chengyuan Zhang$^\dagger$, Lei Zhu$^\dagger$, Weiren Yu$^\ddagger$, Jun Long$^\dagger$, Fang Huang$^\dagger$,  Hongbo Zhao$^{\natural}$}
%


\institute{$^\dagger$ School of Information Science, Central South University, PR China\\
$^{\ddagger}$ School of Engineering and Applied Science, Aston University, United Kingdom\\
$^{\natural}$ School of Minerals Processing and Bioengineering, Central South University, PR China\\
\mailsd\\
}

%
%

\toctitle{Lecture Notes in Computer Science}
\tocauthor{Authors' Instructions}
\maketitle


\vspace{-2mm}
\begin{abstract}
\emph{Massive amount of data that are geo-tagged and associated with text information are being generated at an unprecedented scale in many emerging applications such as location based services and social networks. Due to their importance, a large body of work has focused on efficiently computing various spatial keyword queries. In this paper,we study the top-$k$ temporal spatial keyword query which considers three important constraints during the search including time, spatial proximity and textual relevance. A novel index structure, namely SSG-tree, to efficiently insert/delete spatio-temporal web objects with high rates. Base on SSG-tree an efficient algorithm is developed to support top-k temporal spatial keyword query. We show via extensive experimentation
with real spatial databases that our method has increased performance
over alternate techniques} .
\keywords{Temporal spatial keyword, high rates, memory}
\end{abstract}

\vspace{-2mm}
\section{Introduction}
\label{intro}

Due to the proliferation of user generated content and geo-equipped
devices, massive amount of microblogs (e.g., tweets, Facebook comments, and
Foursquare check-in's) that contain both text information ~\cite{DBLP:conf/mm/WangLWZZ14}
and geographical location information are being generated at
an unprecedented scale on the Web. For instance, in the GPS navigation system, a
POI (point of interest) is a geographically anchored pushpin
that someone may find useful or interesting, which is usually annotated with textual information (e.g., descriptions
and users' reviews). In social media (e.g., Flickr, Facebook,
FourSquare, and Twitter), a large number of posts and photos are usually associated with a geo-position as well as a short text. In the above applications, a large volume
of spatio-textual objects may continuously arrive with high
speed. For instance, it is reported
that there are about 30 million people sending geo-tagged data
out into the Twitterverse, and 2.2 percentage of the global
tweets (about 4.4 million tweets a day) provide location data
together with the text of their posts\footnote{http://www.futurity.org/tweets-give-info-location}.

In this paper, we aim to take advantage of the combination
of geo-tagged information within microblogs to support temporal spatial keyword
search queries on microblogs, where users are interested
in getting a set of recent microblogs each of which contains all keywords and closet to user's location.
Due to the large
numbers of microblogs that can satisfy the given constraints, we limit the query answer to k microblogs, deemed most relevant to the querying user based on a ranking function $f_{st}$
that combines the time recency and the spatial proximity of
each microblog to the querying user.
%

\begin{example}
\label{exa:skt motivation example}
In Fig.~\ref{fig:skt motivation_example}, suppose there are a set of tweets, each of which is described by the send's interests, the sender's location and creation time of tweet. When a GPS-enabled smartphone user wants to find the most recent tweet who has the same interests as him and closes his location, he may send the local search server two keywords, \textit{swimming} and \textit{gym}. Based on the user's current location derived from the smartphone
, the two query keywords, and the creation time, tweet $o_1$ is returned by the server. Note that although tweet $o_3$ is closer to $Q$ than $o_1$,
it doesn't satisfy the keyword constraint. Tweet $o_4$ satisfies the keyword constraint and spatial closer to $Q$ than $o_1$,
but the creation time of $o_4$ is too early.
\end{example}

\vspace{1mm}
\noindent \textbf{Challenges.} There are three key challenges in efficiently processing
temporal spatial keyword queries over temporal spatial keyword microblogs
streams. Firstly, a massive number of microblogs, typically in the
order of millions, are posted in many applications, and
hence even a small increase in efficiency results in significant
savings. Secondly, the streaming temporal spatial keyword microblogs may continuously arrive in a rapid rate
which also calls for high throughput performance for better
user satisfaction. Thirdly, the above challenges call for relying on \textbf{only} in-memory data structures to index and query incoming microblogs, where memory is a scarce resource.

Based on the above challenges, we first discuss what kinds of techniques should be adopted from different angles. Then, we propose a novel index technique, namely the Segment Signature Grid trees (\ssg for short), to effectively organize continuous temporal spatial keyword microblogs.
In a nutshell, \ssg is essentially a set of Signature \gtrees, each node of which is enriched with the reference to a frequency signature file for the objects contained in the sub-tree rooted at the node. Then, an efficient temporal spatial keyword search algorithm is designed to facilitate the online top-$k$ temporal spatial keyword search. Extensive experiments show that our \ssg based \textbf{TSK} algorithm achieves very substantial improvements over the nature extensions of existing techniques due to strong filtering power.


The rest of this paper is organized as follows.
Section~\ref{subsec:skt preliminaires} formally defines the problem of top $k$ temporal spatial keyword search. We introduce the techniques should be adopted in Section~\ref{subsec:skt motivation}. Section~\ref{sec:skt framwork} presents the framework of \ssg and algorithm. Extensive experiments are reported in Section~\ref{sec:skt experiment}.

\vspace{-2mm}
\section{Preliminaires}
\label{subsec:skt preliminaires}
In this section, we present problem definition and necessary preliminaries of top $k$ temporal spatial keyword Search. Table~\ref{tab:skt notation} below summarizes the mathematical notations used throughout this section.

\begin{table}
	\centering
    \small
	\begin{tabular}{|p{0.27\columnwidth}| p{0.62\columnwidth} |}
		\hline
		\textbf{Notation} & \textbf{Definition} \\ \hline\hline
		~$o(q)$             & s geo-textual object (query)                                \\ \hline
        ~$o.\psi(q.\psi)$   & a set of keywords (terms) used to describe o (query q)                                \\ \hline
	  	~$o.loc(q.loc)$     & location of the object o (query q)                                 \\ \hline
        ~$o.t(q.t)$         & timestamp of the object o (query q)  \\ \hline
		~$\mathcal{V}$      & vocabulary                              \\ \hline
		~$w$                & a keyword (term) in $\mathcal{V}$                                \\ \hline	
    	~$l$  & the number of query keywords in $q.\psi$ \\\hline
        ~$w$   & the number of results should be returned  \\\hline
        ~$m$  & the number of independent uniformly random hash functions \\\hline
        ~$\alpha$  & the preference parameter to balance the spatial proximity and temporal recency \\\hline
        ~$b$                & size of a bit-block of GS                               \\ \hline	
        ~$B$                & the sparse vector size                              \\ \hline
        ~$c$                & the leaf node capacity of \ssg                              \\ \hline
        ~$f_s(o.loc,q.loc)$     & the spatial proximity between o.loc and q.loc                                \\\hline
        ~$f_t(o.t,q.t)$         & the temporal recency between o.t and q.t \\ \hline
        ~$f_{st}(o,q)$         & the spatial-temporal ranking score between o and q\\ \hline
	\end{tabular}
    \vspace{2mm}
    \caption{Notations} \label{tab:skt notation}	
    \vspace{-6mm}
\end{table}

In this section, $\mathcal{O}$ denotes a sequence of incoming stream geo-textual objects. A \textbf{geo-textual object} is a textual message with geo-location and timestamp, such as geo-tagged and check-in tweets. Formally, a geo-textual object o is modeled as $o$ = $<\psi,$ $loc,$ $t>$, where $o.\psi$ denotes a set of distinct keywords (terms) from a vocabulary set $\mathcal{V}$, $o.loc$ represents a geo-location with latitude and longitude, and $o.t$ represented the timestamp of object.

\begin{definition}[\textbf{Top-$k$ Temporal Spatial-Keyword ($TSK$) Query}] \label{def:skt sktq}
A top-$k$ temporal spatial keyword query $q$ is defined as q = $<\psi,$ $loc,$ $t,$ $k>$, where $q.\psi$ is a set of distinct user-specified keywords (terms), $q.loc$ is the query location, $q.t$ is the user submitted timestamp, $k$ is the number of the result user expected.
\end{definition}

\begin{definition} [\textbf{Spatial Proximity $f_s(o.loc,q.loc)$}]
Let $\delta_{max}$ denote the maximal distance in the space, the spatial relevance between
the object $o$ and the query $q$, denoted by $f_s(o.loc,q.loc)$, is defined as $\frac{\delta(q.loc, o.loc)}{\delta_{max}}$.
\end{definition}

\begin{definition} [\textbf{Temporal Recency $f_t(o.t,q.t)$}]
Let $\lambda_{max}$ denote the maximal time difference in the timestamp, the temporal recency between
the object $o$ and the query $q$, denoted by $f_t(o.t,q.t)$, is defined as $\frac{\lambda(o.t,q.t)}{\lambda_{max}}$.
\end{definition}

Based on the spatial proximity and temporal recency
between the query and the object, the \textbf{Spatial-temporal Ranking Score} of an object o regarding the query q can
be defined as follows.

\begin{definition} [\textbf{Spatial-temporal Ranking Score} $f_{st}(o,q)$]
Let $\alpha$ ($0\leq\alpha\leq1$)\footnote{$\alpha$=1 indicates that the user cares only about the spatial proximity of geo-textual objects, $\alpha$=0 gives the k most recent geo-textual objects in dataset} be the preference parameter specifying the trade-off between the spatial proximity and temporal recency, we have
\begin{equation}\label{eq:skt function}
f_{st}(o,q)=\alpha*f_{s}(o.loc,q.loc) + (1-\alpha)f_t(o.t,q.t).
\end{equation}
Note that the objects with the \textbf{small score values} are preferred (i,e., ranked higher).
\end{definition}

\begin{definition}[\textbf{Temporal Spatial Keyword Search}] \label{def:skt skts}
Given a set of geo-textual objects $\mathcal{O}$ and a temporal spatial keyword query $q$, we aim to find the top $k$ geo-textual objects with \textbf{smallest} spatial-temporal score, and each of which contains \textbf{all} of the query keywords.
\end{definition}

In the section hereafter, we abbreviate the geo-textual object and the geo-textual query as $object$ and $query$ respectively, if there is no ambiguity. We assume there is a total order for keywords in $\mathcal{V}$, and the keywords in each query and object are sorted accordingly. For presentation simplicity, we assume $w_i<w_j$ if $i<j$.

\vspace{-2mm}
\section{Motivation}\label{subsec:skt motivation}

Due to massive amount of objects and queries are being generated at an unprecedented scale, it is imperative to devise efficient indexing technique such that high arrival rates of incoming objects can be inserted immediately, expired objects can be deleted from its contents with the approximative rate as insertion, a large number of unpromising objects can be filtered at a cheap cost, and the memory cost should linear to the object size increase. We show that a good indexing mechanism over continuous geo-textual objects should satisfy following three criterion.

\subsection{Efficient Textual Retrieval}\label{sec:skt text}
Existing textual retrieval indexes , which can effectively combine with other spatial or temporal indexes, are mainly falling into one of two categories: inverted index ~\cite{DBLP:journals/tip/WangLWZ17,DBLP:journals/tip/WangLWZZH15,DBLP:journals/tnn/WangZWLZ17} and signature file ~\cite{DBLP:conf/sigir/WangLWZZ15,DBLP:journals/pr/WuWLG18,DBLP:journals/cviu/WuWGHL18,DBLP:conf/edbt/ZhangZZLCW14}. Owing to only the indexes of related keywords have been extracted in inverted index, inverted index excels in query processing efficiency while compared with signature. However, signature has faster insertion speed and utilizes significantly less storage overhead. According to~\cite{TSE1984}, inverted index requires much larger space overhead than signature file ($\approx$10 times), and demands expensive updates of the index when insert a new document, due to many terms of inverted index needs to store more than once and frequently undergo re-organization triggers under intensive information
insertion/updating procedures. Furthermore, the inverted index is also reported to perform poorly for multiple terms queries in ~\cite{DBLP:conf/edbt/FaloutsosJ92}. Obviously, taking the properties of fast update online system into consideration, signature file seems a better choice.

\subsection{Efficient Spatial Partition}\label{sec:skt spatial}
To support high arrival rates of incoming objects, space-partitioning index (e.g., \qdtree~\cite{DBLP:journals/cacm/Gargantini82,DBLP:conf/pakdd/WangLZW14,DBLP:journals/corr/abs-1708-02288,DBLP:conf/icde/ZhangZZL13,DBLP:journals/tkde/ZhangZZL16}, \pyr~\cite{DBLP:conf/pods/ArefS90,DBLP:conf/ijcai/WangZWLFP16,DBLP:journals/pr/WuWGL18}, and \gst~\cite{DBLP:conf/sigmod/PapadiasMH05,DBLP:journals/tip/WangLWZ17,DBLP:conf/cikm/WangLZ13}) is more famous than object-partitioning index (e.g., \rtree). As stressed in ~\cite{DBLP:conf/icde/MagdyMENH14}, space-partitioning index is more suitable to high update system because of its disjoint space decomposition policy, while the shape of object-partitioning index is highly affected by the rate and order of incoming data, which may trigger a large number of node splitting and merging. Motivated by this, we should adopt space-partitioning index as our proposed spatial partition index.

\subsection{Efficient Temporal Partition}\label{sec:skt temporal}
Regarding the temporal partition techniques, which can combine with other textual or spatial index, are mainly divided into two categories: \logs ~\cite{DBLP:journals/acta/ONeilCGO96} and \window ~\cite{AmrSIGSPATIAL2014}. \logs partitions the data into a sequence of indexes with exponentially increasing size, while \window  partitions the data into a sequence of indexes with equal size or with equal time range. Obviously, the performance of \logs is better than \window if top-$k$ results can be find in most recent data. However, a sharp drop will be met if the top-$k$ results can be find in most recent data, due to its exponent increase size partition. Furthermore, the insertion cost of \logs will significant increase while combining with other index. Finally, the deletion cost of \logs is always higher than that of \window, due to the deletion operation can only occur at the oldest index. Motivated by the above reasons, our proposed temporal partition strategy should fall into \window.

\vspace{-2mm}
\section{SSG-TREE FRAMEWORK}
\label{sec:skt framwork}

Based on the above motivations, in this section, we present a segment signature grid trees (\ssg for short) that supports update at high arrive rate and provides the following required functions for geo-textual object search and ranking: \MyRoman{1})\textbf{textual filtering}: all the textually irrelevant trees, nodes and objects have to be discarded as early as possible to cut down the search cost; \MyRoman{2})\textbf{spatial filtering}: all the spatially irrelevant nodes have to be filtered out as early as possible to shrink the search space; \MyRoman{3})\textbf{temporal filtering}: all the spatially irrelevant trees, nodes and objects have to be accessed as late as possible to follow the chronological order; and  \MyRoman{4})\textbf{relevance computation and ranking}: since only the top-k objects are returned and $k$ is expected to be much smaller than the total number of match objects, it is desirable to have an incremental search process that integrates the computation of the joint relevance, and object ranking seamlessly so that the search process can stop as soon as the top-$k$ objects are identified.


Below, we first introduce frequency signature to support keyword filtering in section ~\ref{sec:skt f signature}. Section ~\ref{sec:skt grid tree} presents the gird tree for spatial partition. Detail data structure and search algorithm are depicted in section ~\ref{sec:skt structure} and  section ~\ref{sec:skt query} respectively.

\subsection{Frequency Signature}\label{sec:skt f signature}
The traditional superimposed coding signature is widely used in many off-line indexes such as \st~\cite{DBLP:dblp_conf/sigir/Deppisch86}, tow-level superimposed coding~\cite{DBLP:journals/tkde/LeeKP95,NNLS2018,TC2018}, \irtree~\cite{Ian08ICDE} etc. It is well known that the frequency of keyword occurrence in large texts follows Zipf's law. However, the traditional superimposed coding signature does not differentiate the frequencies of different keywords in the dataset. Hence, in this subsection, we study frequency superimposed coding signature based on the keywords's frequency. In many applications, keyword frequencies are estimated or collected with historical data. Statistics of such information are maintained, especially for the high frequency keywords. Such data is useful in optimizing the optimal configuration of superimposed coding signature. The performance improvement is remarkable even with rough estimations of keyword frequencies.

Same as the traditional superimposed coding, the frequency superimposed coding signature also uses $k$ independent uniformly random hash functions map an $n$-terms set $\mathcal{W}$=\{$w_1$, $\cdots$, $w_n$\} into a $B$-bit array. But the major difference is that instead of hashing all terms in range [1,$B$], we divide the $B$-bit array into a set of different size frequency blocks, and hash the terms into different blocks based on their term frequency. More specifically, based on the distribution of term frequency, we partition the terms into different frequency block by a series of frequency thresholds. Assume the aggregate frequency of the entire historical data is $\xi$, the bit array is divided into $u$ frequency blocks, and the aggregate frequency of the $i$-th frequency block is $\xi_i$. Hence, the size of the $i$-th block can be calculated by $\frac{\xi_i}{\xi}\times B$. Then, for each term in $n$-terms set, we can hash them into different frequency blocks. Due to the terms have similar frequency are partitioned into the same blocks and the blocks allocated to high frequency terms have been assign more bits and less terms, the frequency superimposed coding signature effectively avoid the interference from low frequency terms.
\subsection{Grid-Tree}\label{sec:skt grid tree}
The existing space-partitioning techniques are mainly falling into two categories: Grid structure and \qdtree. However, both of them have their own limitations while combined with signature file. The Grid structure is insensitive to system update, but it is hard to decide its granularity. High granularity will cause massive amount of node signature, which will lead to tremendous memory cost. Low granularity will result in a great deal of fat leaf nodes that contain several thousand of objects, which will significantly reduce query performance. Different as Grid structure, \qdtree can dynamically adjust granularity according to object distribution to achieve balance allocation after sacrificing partial update efficiency. Because each internal node of \qdtree has exactly four children, for the leaf node, it is easy to satisfy the leaf node capacity constraint, trigger node split, and complicated object redistribution. Thus, instead of splitting into four children nodes in \qdtree, we partition node into a set of grids. More specifically, the \gtree partitions the spatial node into $n^{2}$  equal non-overlapping girds, where $n\geq 2$, to delay the time of redistribution and avoid continual node split. Evidently, the spatial queries algorithms that can be applied on \qdtree can easily be applied on the \gtree\footnote{Obviously, if $n$ equals two, the \gtree degrades to \qdtree}.

In order to support efficient geo-textual object search, the \ssg clusters a set of geo-textual objects into a series of continual signature \gtree, which cluster the objects into disjointed subsets of nodes and abstracts them in various granularities. By doing so, it capacitates the pruning of those (textually, spatially or temporally) irrelevant subsets or trees. The efficiency of \ssg depends on its pruning power is highly related to the effectiveness of object clustering and the search algorithms. Our \ssg clusters spatially and temporally close objects together and carries textual information in its node signatures.


\ssg is essentially a set of Signature \gtrees, each node of which is enriched with the reference to a signature file for the objects contained in the sub-tree rooted at the node. In particular, each node of an \ssg contains all spatial, temporal, and keyword information; the first is in the form of a rectangle, the second is in the form of timestamp, and the last is in the form of a signature.

More formally, the leaf node of \esg has the form ($nSig$, $r$, $t$, $oSig$). $oSig$ refers to a set of signatures created by the objects of current node, $nSig$ is the \textbf{OR-ing} of all signature in $gSig$, $r$ is the area covered by current node, and $t$ is the latest timestamp aggregated from the objects. An inner node has the form ($nSig$, $r$, $t$, $cp$). $cp$ are the address of the children nodes, $nSig$ is the \textbf{OR-ing} of all the signatures of its children, $r$ is the area covered by current node, and $t$ is the latest timestamp aggregated from its children nodes. To simplify the following presentation we degrade the \gtree to its special case \qdtree in the example.

\subsection{Processing of TSK queries}\label{sec:skt query}

\begin{algorithm}
\begin{algorithmic}[1]
\footnotesize
\caption{\bf TSK Search($q$, $k$, $\mathcal{I}$) }
\label{alg:skt search}
\INPUT $q~:$ the spatial-keyword temporal query, $k~:$ the number of object return,
   $\mathcal{I}~:$ current \ssg index
\OUTPUT $\mathcal{R}:$ top-$k$ query result results

\STATE $\mathcal{R}:= \emptyset$; $\mathcal{H} = \emptyset$, $\lambda_{max} = \infty$

\STATE $\mathcal{H}\leftarrow$ new a min first heap

\STATE build frequency signature for query
\label{alg:skt_build fsig}
\STATE $\mathcal{H}$.Enqueue($\mathcal{I}.root, MIND_{st}(q,\mathcal{I}.root)$)
\label{alg:skt push root}

\WHILE{ $\mathcal{H} \not = \emptyset$ }
     \label{alg:skt_loop_s}
     \STATE $e \leftarrow$ the node popped from $\mathcal{H}$
     \IF{$e$ is a leaf node}
        \label{alg:skt_leaf s}

            \FOR{ each object $o$ in node $e$}
                \IF{$o$ passed the signature test \textbf{AND} $f_{st}(q, o) \leq \lambda_{max}$}
                            \STATE $\lambda_{max} \leftarrow f_{st}(q, o)$
                            \STATE update $\mathcal{R}$ by $(o, f_{st}(q, o))$
                             \label{alg:skt_leaf e}
                \ENDIF
            \ENDFOR

     \ELSE
        \FOR{ each child $e'$ in node $e$}
            \label{alg:skt_non-leaf s}
            \IF{$e'$  passed the signature test \textbf{AND} $MIND_{st}(q, e') \leq \lambda_{max}$}
               \STATE $\mathcal{H}$.Enqueue($e',  MIND_{st}(q, e')$)
               \label{alg:skt_non-leaf e}
            \ENDIF
        \ENDFOR
     \ENDIF

    \STATE process the root node of next \ssg
    \label{alg:skt_process next}
\ENDWHILE
\RETURN{$\mathcal{R}$}
\label{alg:skt_return}
\end{algorithmic}
\end{algorithm}

We proceed to present an important metric, the minimum spatial-temporal distance $MIND_{st}$, which will be used in the query processing. Given a query $q$ and a node $N$ in the \esg, the metric $MIND_{st}$ offers a lower bound on the actual spatial-temporal distance between query $q$ and the objects enclosed in the rectangle of node $N$. This bound can be used to order and efficiently prune the paths of the search space in the \esg.

\begin{definition} [$MIND_{st}(q,N)$]
The distance of a query point $q$ from a node $N$ in
the \esg, denoted as $MIND_{st}(q,N)$, is defined as follows:
\begin{equation}\label{eq:skt mindst}
\begin{aligned}
MIND_{st}(q,N)=\alpha*\frac{MIND_{s}(q.loc, N.r)}{\delta_{max}}+ \\
(1-\alpha)*\frac{MIND_{t}(q.t, N.t)}{\lambda_{max}}
\end{aligned}
\end{equation}
where $\alpha$, $\delta_{max}$, and $\lambda_{max}$ are the same as in Equation \ref{eq:skt function};
$MIND_{s}(q.loc, N.r)$ is the minimum Euclidian distance between $q.loc$ and $N.r$,
$MIND_{t}(q.t, N.t)$ is the minimum time difference between $q.t$ and $N.t$.
\end{definition}

A salient feature of the proposed \esg structure is that
it inherits the nice properties of the \qdtree for query processing.

\begin{theorem}\label{lemma:skt point_prune}
Given a query point $q$, a node $N$, and a set of objects $\mathcal{O}$ in node $N$, for any $o\in\mathcal{O}$, we have $f_{st}(q,N)\leq DIST_{st}(q,o)$.
\end{theorem}
\begin{proof}
Since object $o$ is enclosed in the rectangle of node $N$, the
minimum Euclidian distance between $q.loc$ and $N.r$ is
no larger than the Euclidian distance between $q.loc$ and $o.loc$:

\begin{equation*}
MISD_{S}(q.loc, N.r)\leq f_{s}(q.loc, o.loc)
\end{equation*}
For each timestamp $t$, $N.t$ is the maximum value $\mathcal{O}.t$ of all the object in node $N$. Hence:
\begin{equation*}
MISD_{t}(q.loc, N.r)\leq f_{t}(q.loc, o.loc)
\end{equation*}
According to Equation \ref{eq:skt function} and Equation \ref{eq:skt mindst}, we obtain:
\begin{equation*}
MIND_{st}(q,N)\leq f_{st}(q,o)
\end{equation*}
thus completing the proof.
\end{proof}

When searching the \esg for the $k$ objects nearest to a
query $q$, one must decide at each visited node of the \esg
which entry to search first. Metric $MIND_{ST}$ offers an approximation
of the spatial-temporal ranking score to every entry in the node and, therefore, can
be used to direct the search. Note that only node satisfied the constraint of query keywords need to be loaded into memory
and compute $MIND_{ST}$.

To process \textbf{TSK} queries with \esg framework, we
exploit the best-first traversal algorithm for retrieving
the top-k objects. With the best-first traversal algorithm, a priority
queue is used to keep track of the nodes and objects that have yet to
be visited. The values of $f_{st}$ and $MIND_{st}$ are used as the keys
of objects and nodes, respectively.

When deciding which node to visit next, the algorithm picks the
node $N$ with the smallest $MIND_{st}(q, N)$ value in the set of all
nodes that have yet to be visited. The algorithm terminates when $k$
nearest objects (ranked according to Equation \ref{eq:skt function}) have been found.

Algorithm~\ref{alg:skt search} illustrates the details of the \ssg based \textbf{TSK} query. A minimum heap $\mathcal{H}$ is employed to keep the \gtree's nodes where the key of a node is its minimal spatial-temporal ranking score. For the input query, we calculate its frequency signature in Line~\ref{alg:skt_build fsig}. In Line ~\ref{alg:skt push root}, we find out the root node of current time interval, calculate the minimal spatial-temporal ranking score for the root node, and then pushed the root node into the $\mathcal{H}$. The the algorithm executes the while loop (Line~\ref{alg:skt_loop_s}-\ref{alg:skt_process next})until the top-$k$ results are ultimately reported in Line~\ref{alg:skt_return}.


In each iteration, the top entry $e$ with minimum spatial-temporal ranking score is popped from $\mathcal{H}$. When the popped node $e$ is a leaf node(Line~\ref{alg:skt_leaf s}), for each signature in node $e$, we will iterator  extract the objects that satisfy query constraint and check whether its spatial-temporal ranking score is less than $\lambda_{max}$. If its score is not larger than $\lambda_{max}$, we push $o$ into result set and add update $\lambda_{max}$. When the popped node $e$ is a non-leaf node(Line~\ref{alg:skt_non-leaf s}), a child node $e^{'}$ of $e$ will be pushed to $\mathcal{H}$ if it can pass the query signature test and the minimal spatial-temporal ranking score between $e^{'}$ and $q$, denoted by $MIND_{st}(q, e^{'}$, is not larger than $\lambda_{max}$ (Line~\ref{alg:skt_non-leaf s}-~\ref{alg:skt_non-leaf e}). We process the root node of next interval in Line~\ref{alg:skt_process next}. The algorithm terminates when $\mathcal{H}$ is empty and the results are kept in $\mathcal{R}$.

\newcommand{\ifq}{IFQ\xspace}
\newcommand{\sifq}{SIFQ\xspace}
\newcommand{\essg}{SSG\xspace}
\newcommand{\efssg}{FSSG\xspace}
\newcommand{\ekcssg}{KCSSG\xspace}
\vspace{-2mm}
\section{Experiments}\label{sec:skt experiment}
In this section, we present the results of a comprehensive
performance study to evaluate the effectiveness and efficiency
of our techniques proposed in this section.

\subsection{Baseline Algorithms}\label{sec:baseline}

To the best of our knowledge, no existing work investigating top-$k$ queries on spatial-keyword temporal data.
Hence, for comprehensive performance evaluation, we discuss how to exploit existing techniques
in ~\cite{DBLP:conf/icde/MagdyMENH14} for processing \textbf{TSK} queries. And we develop two baselines by utilizing
 existing index structures, namely IFQ and SIFQ.

\MyRoman{1}) Inverted File plus \qdtree (\textbf{\ifq}). \ifq first employs \qdtree to partition objects
into leaf cells according to their location information. Then, the objects inside each cell are stored in
a reversed chronological order. Finally, for the objects in each leaf cell, we build inverted file for
keyword filtering proposed and recursively construct the inverted file for its ancestral cells.


\MyRoman{2}) Segment Based Inverted File plus Quadtree (\textbf{\sifq}). \sifq is an enhanced version of \ifq. Similarly, it employs \qdtree to partition the objects into different cell, uses reversed chronological order to organize the object list in leaf cell, and build the inverted file for all the cells which contain objects. The major difference between them is that \ifq organizes all the objects in an single quadtrees, but \sifq partitions all the incoming objects into a set of quadtrees or segments by time unit.

\subsection{Experiment Setup}\label{sec:exp setup}

In this section, we implement and evaluate following algorithms.
\begin{itemize}
\item{\textbf{\ifq.}} \ifq based \textbf{TSK} algorithm proposed in baseline algorithm.
\item{\textbf{\sifq.}} Enhanced \ifq by partitioning the objects into a set of time unit in baseline algorithm.
\item{\textbf{\essg.}} \ssg based \textbf{TSK} algorithm proposed in Section ~\ref{sec:skt framwork}.
\end{itemize}

\vspace{1mm}
\noindent \textbf{Dataset.}
All experiments are based on a real-life dataset \textbf{TWEETS}. TWEETS is a real-life dataset collected from Twitter [5], containing 13 million geo-textual tweets from May 2012 to August 2012. The statistics of \textbf{TWEETS} are summarized in Table ~\ref{tab:skt dataset}.
 \begin{table}
	
    \centering
    \begin{tabular}{|c|l|l|l|l|}
	\hline
	\textbf{Property} &  \# of objects & vocabulary & avg. \# of term per obj.  \\
	\hline
	\hline
    \textbf{TWEETS}  & $13.3M$ &  $6.89M$ & $10.05$ \\
    \hline
    \end{tabular}
    \vspace{2mm}
\caption{Dataset Details}
    \vspace{-6mm}
    \label{tab:skt dataset}
\end{table}

\vspace{1mm}
\noindent \textbf{Workload.}
The workload for the \textbf{TSK} query consists of 1000 queries, and the average query response time are
employed to evaluate the performance of the algorithms. The query locations are randomly selected from
the underlying dataset. On the other hand, the likelihood of a keyword $t$ being chosen as query keyword
 is $\frac{freq(t)}{ \sum_{t_i \in \mathcal{V}} freq(t_i)}$ where $freq(t)$ is the term frequency of $t$ in the dataset. The number of query keywords ($l$) varies from $1$ to $5$, the number of results ($k$) grows from $10$ to $50$, and the preference parameter $\alpha$ changes from $0.1$ to $0.9$. By default, $l$, $k$, and $\alpha$  are set to \textbf{3}, \textbf{10}, and \textbf{0.5} respectively. In addition, unless mentioned otherwise, the default value of cell capacity $c$ is $100$, and geo-textual object arrival rate $\tau$ is $4000$. We random select $10\%$ of the geo-textual objects from \textbf{TWEETS} as the historical object workload when
\essg are constructed. We always keep the most recent 5M objects in.


All experiments are implemented in C++. The experiments are conducted on a PC with 2.9GHz Intel Xeon 2 cores CPU and 32GB memory running Red Hat Enterprise Linux. For a fair comparison, we tune the important parameters of the competitor algorithms for their best performance. Particularly, the leaf capacity of all algorithms is set to 100. The partition threshold of \sifq is set to 400000.  Our measures of performance include insertion time, deletion time, storage overhead, and response time. The rest of this section evaluates index maintenance (Section ~\ref{subsec:skt maintaint}), and query processing (Section ~\ref{subsec:skt evaluation}).

\subsection{Index Maintenance}\label{subsec:skt maintaint}

In this subsection, we evaluate the insertion time, deletion time, storage overhead of all the algorithms\footnote{We ignore \ifq's insertion time and deletion time in comparison, due to it cannot meet the current arrive rate.}. Since all the objects are indexed in one single quadtree in \ifq, the insertion and deletion time are longer than the other algorithms. What's worse, the time to process 2000 objects for \ifq is large than one second. Thus, we ignore \ifq's insertion time and deletion time in comparison.

\vspace{0.5mm}
\noindent \textbf{Evaluation on storage overhead.} gives the
performance when varying the number of objects from 5M to 13M. As the number of objects increase, the storage overhead is stable for all algorithms, owing to we only keep the most
recent 5M objects in memory. The performance of signature based algorithms is always better than that of inverted index based algorithms. The storage overhead of signature based algorithms is at most one half of that of inverted index based algorithms. gives the same experiment with varying node capacity from 100 to 500. The storage overhead of algorithms meets a slight decrease, due to more information has been shared by the textual index.

\vspace{-2mm}
\section{Conclusions}\label{concl: con}

To the best of our knowledge, this is the first work to study the problem of top-$k$ continuous temporal spatial keyword queries over streaming temporal spatial keyword microblogs, which has a wide spectrum of application. To tackle with this problem, we propose a novel temporal spatial keyword  partition indexing structure, namely \ssg, efficiently organize a massive number of streaming temporal spatial keyword microblogs such that each incoming query submitted by users can rapidly find out the top-$k$ results. Extensive experiments
demonstrate that our technique achieves a high throughput performance over streaming temporal spatial keyword data.

\textbf{Acknowledgments:} This work was supported in part by the National Natural Science Foundation of China
(61272150, 61379110, 61472450, 61402165, 61702560, S1651002, M1450004), the Key Research Program of Hunan Province(2016JC2018), and project (2016JC2011, 2018JJ3691) of Science and Technology Plan of Hunan Province.

\bibliographystyle{spmpsci}      
\bibliography{ref}
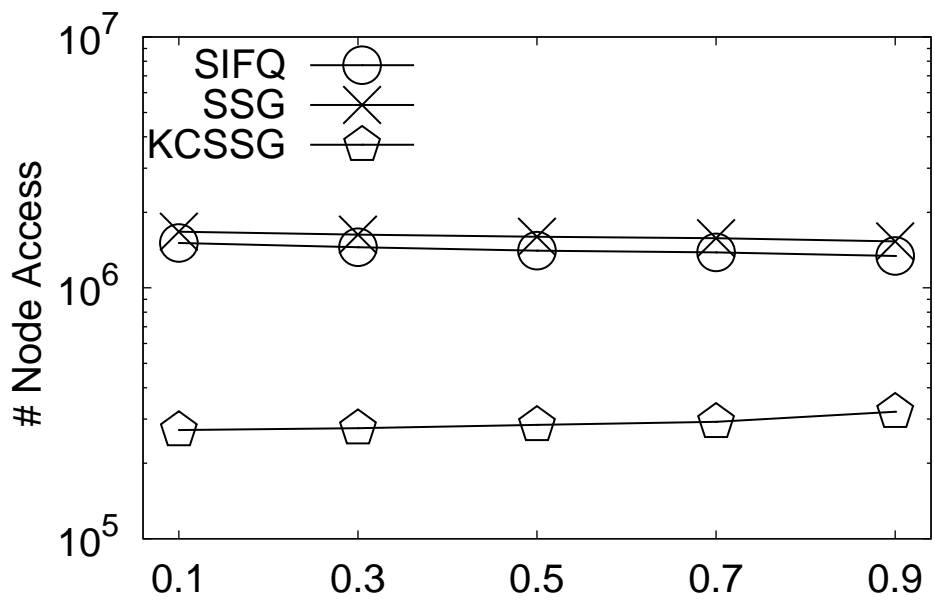
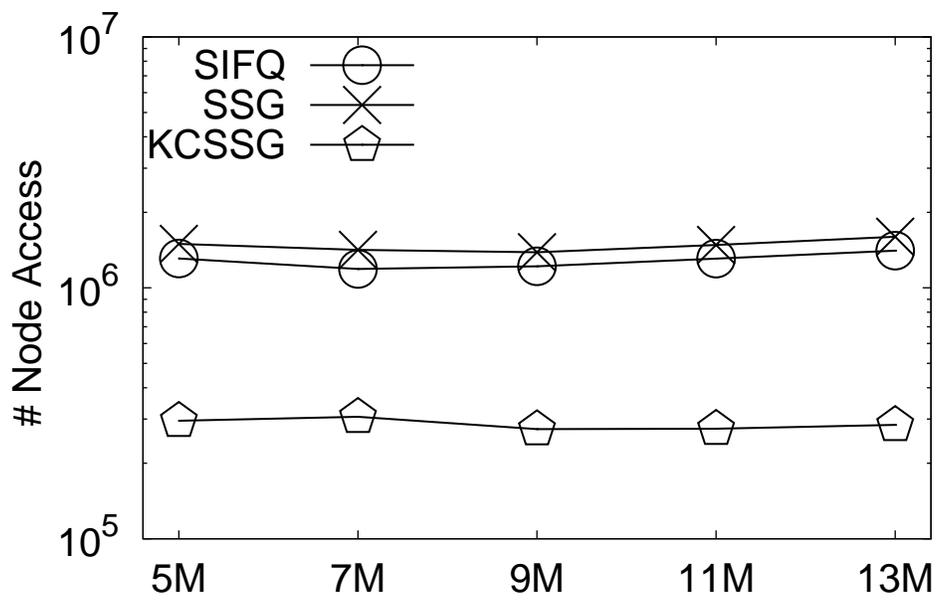
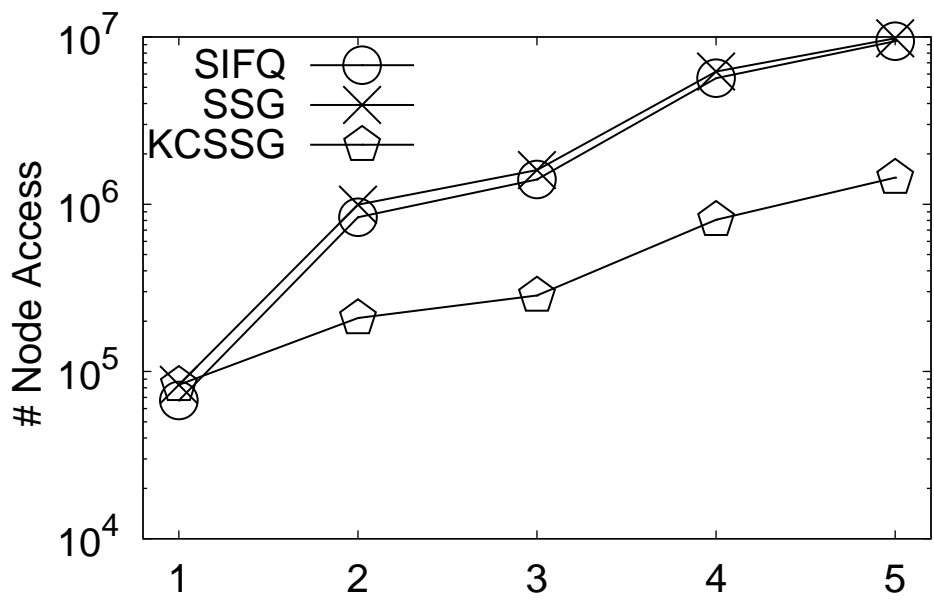
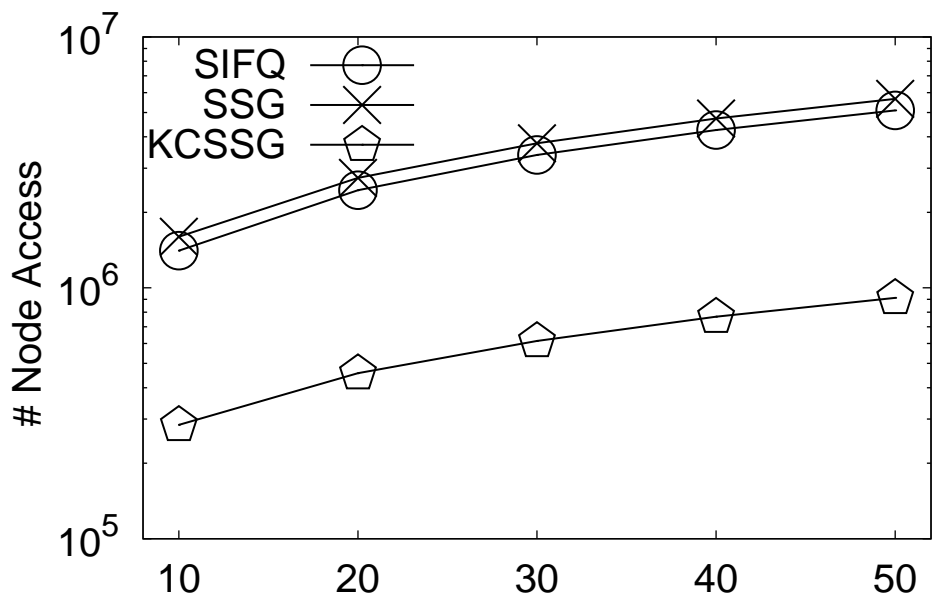
\end{document}